\documentclass[a4paper,english]{lipics}
\newenvironment{myquote}{\list{}{\leftmargin=5pt\rightmargin=0in\topsep=3pt}\item[]}{\endlist}

\usepackage{microtype}

\usepackage{amsthm}
\usepackage{booktabs}
\usepackage{bm}
\usepackage{url}
\usepackage{mathtools}
\usepackage{hyperref}

\usepackage[usenames,dvipsnames]{xcolor}
\usepackage[colorinlistoftodos]{todonotes}

\usepackage{tikz}
\usetikzlibrary{arrows,backgrounds,fit}

\newcommand{\SB}{\{\,}%
\newcommand{\SM}{\;{:}\;}%
\newcommand{\SE}{\,\}}%
\newcommand{\SBs}{\{}%
\newcommand{\SEs}{\}}%

\renewcommand{\P}{\text{\normalfont P}}
\newcommand{\NP}{\text{\normalfont NP}}

\newcommand{\ThetaP}[1]{\ensuremath{\Theta^{\mtext{p}}_{#1}}}

\newcommand{\NN}{\mathbb{N}}

\newcommand{\K}{\mtext{\textsf{K}}}

\newcommand{\mtext}[1]{\text{\normalfont #1}}

\newcommand{\Mod}[1]{\mtext{\textsf{Mod}\ensuremath{(#1)}}}
\newcommand{\Var}[1]{\mtext{\textsf{Var}\ensuremath{(#1)}}}
\newcommand{\Lit}[1]{\mtext{\textsf{Lit}\ensuremath{(#1)}}}

\newcommand{\lin}[0]{\mtext{\textsf{lin}}}
\newcommand{\lintop}[0]{\mtext{\textsf{lin-top}}}
\newcommand{\red}[0]{\mtext{\textsf{red}}}
\newcommand{\green}[0]{\mtext{\textsf{green}}}

\allowdisplaybreaks

  \renewenvironment{proof}{\vspace{-2mm}\begin{pf}}{\qed\end{pf}}

  \theoremstyle{theorem}
  \newtheorem{proposition}[theorem]{Proposition}
  
  \newtheorem{proposition*}[theorem]{Proposition$^{\star}$}
  \newtheorem{theorem*}[theorem]{Theorem$^{\star}$}
  \theoremstyle{remark}
  \newtheorem{observation}[theorem]{Observation}

\renewenvironment{example}{\begin{ex}}{\hfill
    $\dashv$\end{ex}\medskip}

\DeclareRobustCommand{\DE}[3]{#2}

\title{Expressing Linear Orders Requires Exponential-Size DNNFs}
\titlerunning{Expressing Linear Orders Requires Exponential-Size DNNFs}
\author[]{Ronald de Haan}
\affil[]{Institute for Logic, Language and Information\\University of Amsterdam\\\url{me@ronalddehaan.eu}}
\authorrunning{R.~de Haan}
\Copyright{Ronald de Haan}

\serieslogo{}
\volumeinfo
  {}
  {1}
  {}
  {}
  {}
  {1}
\EventShortName{}
\DOI{}

\begin{document}

\maketitle

\begin{abstract}
We show that any DNNF circuit that expresses the set of linear orders over
a set of~$n$ candidates must be of size~$2^{\Omega(n)}$.
Moreover, we show that there exist DNNF circuits of size~$2^{O(n)}$
expressing linear orders over~$n$ candidates.
\end{abstract}

\section{Introduction}

This report considers a technical question that plays a role in the investigation
of the expressivity and efficiency of different knowledge representation
formalisms for social choice applications.
In particular, we consider the formalism of Boolean circuits in
\emph{Decomposable Negation Normal Form (DNNF)}
(or DNNF circuits).
This is a formalism that has been studied in the setting of
\emph{knowledge compilation}
and that enjoys many positive algorithmic properties
\cite{DarwicheMarquis02}.
We study the question whether the formalism of DNNF circuits
can be used to express linear preferences in an efficient and compact
way.
We will answer this question in the negative:
any DNNF circuit that expresses the set of linear orders over
a set of~$n$ candidates must be of size~$2^{\Omega(n)}$---%
that is, of size exponential in~$n$.

This result is of relevance for the field of computational social choice.
DNNF circuits can be used as a representation formalism for the
framework of \emph{Judgment Aggregation},
allowing many judgment aggregation procedures to be carried out
efficiently \cite{DeHaan18a}.
The result of this paper shows that simulating preference
aggregation in judgment aggregation using DNNF circuits
results in an exponential overhead---thereby preventing this
approach to lead to efficient algorithms.

Several results known from the literature together already
implied a (weaker) lower bound on the size of DNNF circuits
expressing linear orders:
(1)~computational intractability results (\ThetaP{2}-completeness)
for the Kemeny rule in preference aggregation
\cite{HemaspaandraSpakowskiVogel05},
(2)~polynomial-time solvability for the Kemeny procedure
in judgment aggregation when the integrity constraint is expressed
using a DNNF circuit \cite{DeHaan18a},
and (3)~the fact that the Kemeny rule in preference aggregation
can be simulated by the Kemeny procedure in judgment aggregation
when the integrity constraint expresses linear orders over the
set of candidates \cite{Endriss18}.
These three results together imply the (conditional) result
that DNNF circuits expressing linear orders must be of
superpolynomial size, unless~$\P = \NP$.
The result in this report strengthens this results in two ways:
(A)~our result is unconditional (i.e., it does not rely on the
assumption that~$\P \neq \NP$), and
(B)~our result gives an exponential lower bound,
rather than a mere superpolynomial lower bound.

\section{Preliminaries}

We introduce several technical notions that play a role
in the result of this report.

\subsection{Propositional Logic}

Propositional formulas are constructed from propositional variables
using the Boolean operators~$\wedge,\vee,\rightarrow$, and~$\neg$.
A \emph{literal} is a propositional variable~$x$ (a \emph{positive literal})
or a negated variable~$\neg x$ (a \emph{negative literal}).
A \emph{clause} is a finite set of literals,
not containing a complementary pair~$x$,~$\neg x$,
and is interpreted as the disjunction of these literals.
A formula in \emph{conjunctive normal form (CNF)}
is a finite set of clauses, interpreted as the conjunction
of these clauses.

For a propositional formula~$\varphi$,~$\Var{\varphi}$ denotes
the set of all variables occurring in~$\varphi$.
Moreover, for a set~$X$ of variables,~$\Lit{X}$ denotes
the set of all literals over variables in~$X$,
i.e.,~$\Lit{X} = \SB x, \neg x \SM x \in X \SE$.

We use the standard notion of \emph{(truth)
assignments}~$\alpha : \Var{\varphi} \rightarrow \SBs 0,1 \SEs$
for Boolean formulas and \emph{truth} of a formula
under such an assignment.
For any formula~$\varphi$ and any truth assignment~$\alpha$,
we let~$\varphi[\alpha]$ denote the formula obtained from~$\varphi$
by instantiating variables~$s$ in the domain of~$\alpha$ with~$\alpha(x)$
and simplifying the formula accordingly.
By a slight abuse of notation,
if~$\alpha$ is defined on all~$\Var{\varphi}$,
we let~$\varphi[\alpha]$
denote the truth value of~$\varphi$ under~$\alpha$.

\subsection{Boolean Functions Expressing Linear Orders}

We consider a family of Boolean functions whose models correspond
to linear orders over a set of~$n$ elements.
A linear order over a set~$A$ is a binary relation~${\prec} \subseteq A \times A$
that is irreflexive (for all~$a \in A$,~$a \not\prec a$),
transitive (for all~$a,b,c \in A$, if~$a \prec b$ and~$b \prec c$ then~$a \prec c$)
and complete (for all~$a,b \in A$, either~$a \prec b$ or~$b \prec a$).

\begin{definition}
The family of
propositional functions~$\SBs \lin_n \SEs_{n \in \NN}$
is defined in such a way that for each~$n \in \NN$ the function~$\lin_n$
is the Boolean function over the variables~$\SB x_{i,j} \SM i,j \in [n], i < j \SE$
such that~$f[\alpha]$ is true if and only if there exists
some linear order~${\prec} \subseteq [n] \times [n]$
such that for each~$i,j \in [n]$ with~$i < j$
it holds that~$\alpha(x_{i,j}) = 1$ if and only if~$i \prec j$
(and thus that~$\alpha(x_{i,j}) = 0$ if and only if~$j \prec i$).
\end{definition}

That is, for each~$n \in \NN$, the models of~$\lin_n$
are in one-to-one correspondence with linear orders~$\prec$ over~$[n]$,
where~$x_{i,j}$ is true if and only if~$i \prec j$.

\begin{observation}
The family~$\SBs \lin_n \SEs_{n \in \NN}$ of Boolean functions
can be expressed with a family~$\SBs \varphi_n \SEs_{n \in \NN}$
of 3CNF formulas that is of size~$O(n^3)$.
\end{observation}
\begin{proof}
Take an arbitrary~$n \in \NN$.
Define~$\varphi_n$ as follows:
\[ \varphi_n = \bigwedge\limits_{i,j,k \in [n] \atop i \neq j, i \neq k, j \neq k}
  (\neg x_{i,j} \vee \neg x_{j,k} \vee x_{i,k}), \]
where~$x_{j,i}$ denotes~$\neg x_{i,j}$ for each~$i,j \in [n]$ with~$j < i$.
\end{proof}

%

\subsection{DNNF Circuits}

Boolean circuits in \emph{Decomposable Negation Normal Form}
(or \emph{DNNF circuits})
are a particular class of Boolean circuits in
\emph{Negation Normal Form (NNF)}
\cite{DarwicheMarquis02}.
A Boolean circuit~$C$ in NNF is a direct acyclic graph with a single root
(a node with no ingoing edges) where each leaf is labelled
with~$\top$,~$\bot$,~$x$ or~$\neg x$ for a propositional variable~$x$,
and where each internal node is labelled with~$\wedge$ or~$\vee$.
(An arc in the graph from~$N_1$ to~$N_2$ indicates that~$N_2$ is a child
node of~$N_1$.)
The set of propositional variables occurring in~$C$ is denoted
by~$\Var{C}$.
For any truth assignment~$\alpha : \Var{C} \rightarrow \SBs 0,1 \SEs$, we define
the truth value~$C[\alpha]$ assigned to~$C$ by~$\alpha$ in the usual
way, i.e., each node is assigned a truth value based on its label and
the truth value assigned to its children, and the truth value assigned to~$C$
is the truth value assigned to the root of the circuit.
DNNF circuits are Boolean circuits in NNF that satisfy the additional
property of decomposability.
A circuit~$C$ is \emph{decomposable} if for each conjunction in the circuit,
the conjuncts do not share variables.
That is, for each node~$d$ in~$C$ that is labelled with~$\wedge$ and for any
two children~$d_1,d_2$ of this node, it holds that~$\Var{C_1} \cap \Var{C_2} = \emptyset$,
where~$C_1,C_2$ are the subcircuits of~$C$ that have~$d_1,d_2$ as root,
respectively.

\subsection{Rectangle Covers}

We introduce the notion of rectangle covers
(see, e.g.,~\cite{BovaCapelliMengelSlivovsky16}).
Let~$X$ be a finite set of propositional variables.
A \emph{partition} of~$X$ is a sequence of pairwise disjoint
subsets of~$X$ whose union is~$X$.
A partition~$(X_1,X_2)$ of~$X$ is called \emph{balanced}
if~$|X|/3 \leq \min(|X_1|,|X_2|)$.

Let~$(X_1,X_2)$ be a partition of~$X$.
For truth assignments~$b_1 : X_1 \rightarrow \SBs 0,1 \SEs$
and~$b_2 : X_2 \rightarrow \SBs 0,1 \SEs$
we let~$b_1 \cup b_2 : X \rightarrow \SBs 0,1 \SEs$
denote the (unique) truth assignment that agrees with both~$b_1$
and~$b_2$.
For sets~$B_1 \subseteq \SBs 0,1 \SEs^{X_1}$
and~$B_2 \subseteq \SBs 0,1 \SEs^{X_2}$ of truth assignments
to~$X_1$ and~$X_2$, respectively, we
let~$B_1 \times B_2 = \SB b_1 \cup b_2 \SM b_1 \in B_1, b_2 \in B_2 \SE$.

A \emph{(combinatorial) rectangle} over~$X$ is a
set~$R \subseteq \SBs 0,1 \SEs^{X}$ of truth assignments to the
variables~$X$ such that there exists an underlying partition~$(X_1,X_2)$
of~$X$ and sets~$R_1 \subseteq \SBs 0,1 \SEs^{X_1}$
and~$R_2 \subseteq \SBs 0,1 \SEs^{X_2}$ of truth assignments
to~$X_1$ and~$X_2$, respectively, such that~$R = R_1 \times R_2$.
A rectangle is \emph{balanced} if its underlying partition
is balanced.

Let~$f : X \rightarrow \SBs 0,1 \SEs$ be a Boolean function over the
set~$X$ of variables.
A finite set~$\SBs R^1,\dotsc,R^t \SEs$ of rectangles over~$X$
is a \emph{rectangle cover} of~$f$ if:
\[ \Mod{f} = \SB (\alpha : X \rightarrow \SBs 0,1 \SEs) \SM f[\alpha] = 1 \SE
= \bigcup\limits_{1 \leq i \leq t} R^i. \]
A rectangle cover is called \emph{balanced} if each rectangle in
the cover is balanced.

\section{The Result}

In this section, we establish the following result.

\begin{theorem}
\label{thm:linear-dnnf-exponential}
Every family~$\SBs \Gamma_n \SEs_{n \in \NN}$
of DNNF circuits that expresses the family~$\SBs \lin_n \SEs_{n \in \NN}$
of Boolean functions expressing linear orders
is of size~$2^{\Omega(n)}$.
\end{theorem}

We will use two lemmas to establish this result.

\begin{lemma}
\label{lem:1}
Let~$\K_n = (V,E)$ be the complete graph on~$n$ vertices.
Moreover, let~$\rho : E \rightarrow \SBs \red,\green \SEs$ be
an edge coloring of~$\K_n$ (using two colors)
that colors at least~$n/3$ edges
with~$\red$ and at least~$n/3$ edges with~$\green$.
Then there exists at least~$\frac{1}{100}n$ vertices~$v \in V$
such that~$\rho$ colors at least~$\frac{1}{100}n$ edges
that are adjacent to~$v$ with~$\red$
and colors at least~$\frac{1}{100}n$ edges
that are adjacent to~$v$ with~$\green$.
\end{lemma}
\begin{proof}
Take some~$n \in \NN$, and
let~$\K_n = (V,E)$ be the complete graph on~$n$ vertices.
Without loss of generality, assume that~$n \geq 2$.
Moreover, let~$\rho : E \rightarrow \SBs \red,\green \SEs$ be
an edge coloring of~$\K_n$ (using two colors)
that colors at least~$|E|/3$ edges
with~$\red$ and at least~$|E|/3$ edges with~$\green$.
Let~$V_{\red} \subseteq V$ be the set of vertices in~$V$
that have more adjacent edges colored with~$\red$ by~$\rho$,
and similarly let~$V_{\green} \subseteq V$ be the set of vertices in~$V$
that have more adjacent edges colored with~$\green$ by~$\rho$.

We distinguish several cases:
(i)~there are at least~$\frac{1}{100}n$ vertices in~$V_{\red}$
that each have at least~$\frac{1}{100}n$ adjacent green edges,
(ii)~there are at least~$\frac{1}{100}n$ vertices in~$V_{\green}$
that each have at least~$\frac{1}{100}n$ adjacent red edges,
or~(iii) neither~(i) nor~(ii) is the case---i.e.,
there are less than~$\frac{1}{100}n$ vertices in~$V_{\red}$
that each have at least~$\frac{1}{100}n$ adjacent green edges
and there are less than~$\frac{1}{100}n$ vertices in~$V_{\green}$
that each have at least~$\frac{1}{100}n$ adjacent red edges.

Firstly consider case~(i).
Then, since all vertices in~$V_{\red}$ have at least~$\frac{1}{2}(n-1)
\geq \frac{1}{120}n$
adjacent green edges, we know that there are at least~$\frac{1}{120}n$
vertices with at least~$\frac{1}{120}n$ adjacent green edges
and at least~$\frac{1}{120}n$ adjacent red edges.
Thus the result follows.
The case of~(ii) is entirely similar.

Now we turn to the case of~(iii).
We show that~$|V_{\red}| \leq \frac{3}{4}n$.
Suppose, to derive a contradiction, that~$|V_{\red}| > \frac{3}{4}n$.
We now doubly count all edges.
The number of green edges (counted doubly) is
then less than~$\frac{1}{100}n \cdot (n-1) + \frac{296}{400}n \cdot \frac{1}{100}n +
\frac{1}{4}n \cdot (n-1) = \frac{10696}{40000}n^2 - \frac{26}{100}n
\leq \frac{11}{40}n^2 - \frac{1}{4}n$.
However, we know that the number of green edges (counted doubly)
is at least~$2 \cdot \frac{1}{3} \cdot \binom{n}{2} =
\frac{1}{3}n^2 - \frac{1}{3}n$.
Since~$n \geq 2$,~$\frac{11}{40}n^2 - \frac{1}{4}n
\leq \frac{1}{3}n^2 - \frac{1}{3}n$, which leads to a contradiction.
Thus, we can conclude that~$|V_{\red}| \leq \frac{3}{4}n$.
By an entirely similar argument, we have
that~$|V_{\green}| \leq \frac{3}{4}n$.
Thus, also,~$|V_{\red}| \geq \frac{1}{4}n$
and~$|V_{\green}| \geq \frac{1}{4}n$.

Now, without loss of generality we may assume
that~$|V_{\red}| \geq |V_{\green}|$---%
the case for~$|V_{\red}| \leq |V_{\green}|$ is entirely analogous.
Thus~$|V_{\green}| \leq \frac{1}{2}n$.
We now split the edges in~$E$ into three groups:
(1)~edges between two vertices in~$V_{\red}$,
(2)~edges between two vertices in~$V_{\green}$, and
(3)~edges between a vertex in~$V_{\red}$
and a vertex in~$V_{\green}$.
Again, we will doubly count all the edges in~$E$.
We will count the number of green edges.
There are at most~$\frac{1}{100}n \cdot (n-1) + \frac{296}{400}n \cdot \frac{1}{100}n
= \frac{696}{40000}n^2 - \frac{1}{100}n$ green
edges in Group~1 (doubly counted).
There are at most~$\frac{n}{2}(\frac{n}{2}-1)$ green edges
in Group~2 (doubly counted).
Since there are at least~$2 \cdot \frac{1}{3} \cdot \binom{n}{2} =
\frac{1}{3}n^2 - \frac{1}{3}n$ green edges in all groups
combined (doubly counted),
this means that there must be at
least~$\frac{2}{30}n^2 + \frac{1}{6}n$ green edges in Group~3
(doubly counted).
This means that there must be at least~$\frac{1}{30}n^2$
green edges in Group~3 (singly counted).
Then, since~$|V_{\red}| \leq \frac{3}{4}n$,
we know that for the nodes in~$V_{\red}$
the average number of adjacent green edges
is at least~$\frac{4}{90}n$.
Then, since each node in~$V_{\red}$ is adjacent
to at least~$0$ and at most~$n-1$ green edges,
and since~$|V_{\red}| \leq \frac{3}{4}n$,
there must be at least~$\frac{12}{360}n$ nodes in~$V_{\red}$
that are adjacent to at least~$\frac{4}{90}n$ green edges.
Thus, the result follows.
\end{proof}

\begin{lemma}
\label{lem:2}
Let~$n \in \NN$ 
and let~$R$ be a balanced
rectangle over~$X = \SB x_{i,j} \SM i,j \in [n], i < j \SE$
such that~$R \subseteq \Mod{\lin_n}$.
Then~$|R| \leq n! / 2^{cn}$,
for~$c = \frac{1}{5200}$.
\end{lemma}
\begin{proof}
Take an arbitrary~$n \in \NN$ with~$n \geq 10$,
and let~$R$ be a balanced rectangle
over~$X = \SB x_{i,j} \SM i,j \in [n], i < j \SE$
such that~$R \subseteq \Mod{\lin_n}$.
Let~$(X_1,X_2)$ be the underlying partition for~$R$,
and let~$R_1 \subseteq \SBs 0,1 \SEs^{X_1}$
and~$R_2 \subseteq \SBs 0,1 \SEs^{X_2}$ be the sets of truth
assignments such that~$R = R_1 \times R_2$.
Since~$R$ is balanced, we know that~$|X_1| \geq \frac{1}{3}|X|$
and that~$|X_2| \geq \frac{1}{3}|X|$.

Now consider the complete graph~$\K_n = (V,E)$ on~$n$ vertices,
where~$V = [n]$.
Moreover, let~$\rho : E \rightarrow \SBs \red,\green \SEs$ be
the edge coloring of~$\K_n$ that is defined as follows:
for each~$i,j \in [n]$ with~$i < j$, let~$\rho(\SBs i,j \SEs) = \red$
if~$x_{i,j} \in X_1$, and let~$\rho(\SBs i,j \SEs) = \green$
if~$x_{i,j} \in X_2$.
Since~$(X_1,X_2)$ is balanced, we know that~$\rho$ is
an edge coloring of~$\K_n$ (using only the colors~$\red$ and~$\green$)
that colors at least~$n/3$ edges
with~$\red$ and at least~$n/3$ edges with~$\green$.
Thus, Lemma~\ref{lem:1} applies,
and we can conclude that there are
at least~$\ell \geq \frac{1}{100}n$
vertices~$v_1,\dotsc,v_{\ell} \in V$
such that for each~$i \in [\ell]$
it holds that~$\rho$ colors at least~$\ell$ edges
adjacent to~$v_i$ with~$\red$
and colors at least~$\ell$ edges
adjacent to~$v_i$ with~$\green$.

Therefore in~$\K_n$ there exist at
least~$\frac{1}{200}n$ disjoint triangles that contain
at least one green edge and at least one red edge
(w.r.t.~$\rho$).
Then we also know that there are either
(i)~at least~$\frac{1}{400}n$ disjoint triangles
that contain exactly one red edge,
or (ii)~at least~$\frac{1}{400}n$ disjoint triangles
that contain exactly one green edge.
Without loss of generality, we assume that~(i) is the case---%
the case of~(ii) is entirely analogous.
Let these triangles be~$t_1,\dotsc,t_k$,
for~$k \geq \frac{1}{400}n$,
where for each~$i \in [k]$ it holds that
the vertices in~$t_i$ are~$a_i,b_i,c_i \in V$,
such that the edge~$\SBs a_i,c_i \SEs$ is colored
with~$\red$ by~$\rho$
(and the edges~$\SBs a_i,b_i \SEs$
and~$\SBs b_i,c_i \SEs$ are colored
with~$\green$ by~$\rho$).

Now, we may assume that~$R_1 \neq \emptyset$---%
otherwise~$R = \emptyset$ and the result would follow immediately.
Take some~$\alpha \in R_1$.
That is,~$\alpha : X_1 \rightarrow \SBs 0,1 \SEs$ is a truth assignment
to the variables in~$X_1$.
We will now define several partial truth assignments to the variables
in~$X_2$.
(In the remainder we let~$x_{j,i}$ denote~$\neg x_{i,j}$
for each~$i,j \in [n]$ with~$j<i$.)
For each~$j \in [k]$, we will define
the partial truth assignments~$\beta_{j}^{1},\dotsc,\beta_{j}^{6}$
as follows:
\[ \begin{array}{r l}
  \beta_j^1 = & \SBs
    x_{a_j,b_j} \mapsto 1, x_{b_j,c_j} \mapsto 1, x_{a_j,c_j} \mapsto 1 \SEs, \\ 
  \beta_j^2 = & \SBs
    x_{a_j,b_j} \mapsto 1, x_{b_j,c_j} \mapsto 0, x_{a_j,c_j} \mapsto 1 \SEs, \\ 
  \beta_j^3 = & \SBs
    x_{a_j,b_j} \mapsto 0, x_{b_j,c_j} \mapsto 1, x_{a_j,c_j} \mapsto 1 \SEs, \\ 
  \beta_j^4 = & \SBs
    x_{a_j,b_j} \mapsto 0, x_{b_j,c_j} \mapsto 1, x_{a_j,c_j} \mapsto 0 \SEs, \\ 
  \beta_j^5 = & \SBs
    x_{a_j,b_j} \mapsto 1, x_{b_j,c_j} \mapsto 0, x_{a_j,c_j} \mapsto 0 \SEs,
    \mtext{ and} \\ 
  \beta_j^6 = & \SBs
    x_{a_j,b_j} \mapsto 0, x_{b_j,c_j} \mapsto 0, x_{a_j,c_j} \mapsto 0 \SEs. \\ 
\end{array} \]
Moreover, if~$\alpha$ satisfies~$x_{a_j,c_j}$,
we let~$\beta_j = \beta_j^{6}$,
and if~$\alpha$ satisfies~$\neg x_{a_j,c_j}$,
we let~$\beta_j = \beta_j^{1}$.

Now, take any~$\overline{s} = (s_1,\dotsc,s_k) \in [6]^{k}$.
Then the set of truth assignments~$B_{\overline{s}} \subseteq \Mod{\lin_n}$
is defined to be those~$\gamma \in \Mod{\lin_n}$
that are consistent with~$\beta_j^{s_j}$, for each~$j \in [k]$.
For each~$\overline{s},\overline{s}' = \in [6]^{k}$,
it holds that~$|B_{\overline{s}}| = |B_{\overline{s}'}|$
and~$B_{\overline{s}} \cap B_{\overline{s}'} = \emptyset$.
Therefore, since~$|\Mod{\lin_n}| = n!$,
for each~$\overline{s} \in [6]^{k}$ it holds
that~$|B_{\overline{s}}| = n! \cdot (1/6)^k$.

We claim that there cannot be any~$\gamma \in R$
that agrees with~$\beta_j$, for any~$j \in [k]$.
Suppose that this were not the case, for some~$j \in [k]$.
We know that~$x_{a_j,b_j},x_{b_j,c_j} \in X_2$,
and~$x_{a_j,c_j} \in X_1$. Then there
must be some~$\alpha' \in R_2$ that agrees with~$\beta_j$.
However, we also know that there exists
some~$\alpha \in R_1$ such that~$\alpha(x_{a_j,c_j})$
together with~$\alpha'$ is not in~$\Mod{\lin_n}$
(by our selection of~$\beta_j$).
Therefore, we know that for each~$\gamma \in R$
and for each~$j \in [k]$,
it holds that~$\gamma$ agrees with
exactly one partial truth assignment in~$\SB \beta_j^\ell \SM \ell \in [6] \SE
\setminus \SBs \beta_j \SEs$.
Thus, we can conclude that
there are at most~$n! \cdot (5/6)^{k}$
different truth assignments~$\gamma \in R$.
In other words,~$|R| \leq n! \cdot (5/6)^{k}
= n! \cdot 2^{k \log (5/6)}
= n! / 2^{k \log (6/5)}
\leq n! / 2^{k/13}
\leq n! / 2^{n/5200}$.
\end{proof}

%
%
%

We are now ready to prove Theorem~\ref{thm:linear-dnnf-exponential}.

\begin{proof}[Proof of Theorem~\ref{thm:linear-dnnf-exponential}]
Take an arbitrary~$n \geq 2$.
Moreover, take a DNNF circuit~$C$ expressing~$\lin_n$.
We show that~$C$ is of size at least~$2^{cn}$,
for some constant~$c > 0$.
Let~$u$ denote the size of~$C$.
We know that this implies that
there is a balanced rectangle cover~$\SBs R^1,\dotsc,R^u \SEs$
of~$\lin_n$ of size~$u$ \cite{BovaCapelliMengelSlivovsky16}.
We know that~$|\Mod{\lin_n}| = n!$.
Moreover, by Lemma~\ref{lem:2}, we know that
every rectangle~$R^j$ in the rectangle cover~$\SBs R^1,\dotsc,R^u \SEs$
contains at most~$n! / 2^{cn}$ truth assignments in~$\Mod{\lin_n}$,
for some constant~$c > 0$.
Therefore,~$u \geq 2^{cn}$---%
in other words, the DNNF circuit is of size~$2^{\Omega(n)}$.
\end{proof}

We give an accompanying upper bound,
showing for each~$n \in \NN$ that we can
express~$\lin_n$ using a DNNF circuit
of size~$2^{O(n)}$.

\begin{proposition}
\label{prop:upper-bound}
For each~$n \in \NN$, there is a DNNF circuit
of size~$2^{O(n)}$ expressing~$\lin_n$.
\end{proposition}
\begin{proof}
Take some~$n \in \NN$,
and consider~$\lin_n$
over the variables~$\SB x_{i,j} \SM i,j \in [n], i < j \SE$.
We construct a DNNF circuit expressing~$\lin_n$
as follows.
Let~$T$ be the set of all subsets of~$[n]$---%
that is~$T = \SB S \SM S \subseteq [n] \SE$.
We know that~$|T| = 2^n$.
We introduce nodes~$C_{S}$ for each~$S \in T$,
and nodes~$C_{i,S}$ for each~$S \in T \setminus \SBs [n] \SEs$ and
each~$i \in [n] \setminus S$.
For each~$S \in T \setminus \SBs [n] \SEs$, we let:
\[ C_{S} = \bigvee\limits_{i \in [n] \setminus S} C_{i,S}. \]
For each~$S \in T \setminus \SBs [n] \SEs$ and each~$i \in [n] \setminus S$,
we let:
\[ C_{i,S} = \bigwedge\limits_{j \in [n] \setminus S \atop i \neq j}
  x_{i,j} \wedge C_{S \cup \SBs i \SEs}, \]
where~$x_{i,j}$ denotes~$\neg x_{j,i}$ if~$j < i$.
Finally, we let~$C_{[n]} = \top$.
We then let the node~$C_{\emptyset}$ be the root of the Boolean circuit.
It is straightforward to verify that the circuit
expresses~$\lin_n$.
Moreover, the circuit is in DNNF
because each node~$C_{S}$ contains only
variables~$x_{i,j}$ with~$i,j \in [n] \setminus S$.
\end{proof}

\section{Conclusion \& Future Research}

This report contains the technical result that DNNF circuits expressing
linear orders over~$n$ elements must be of size~$2^{\Omega(n)}$.
Moreover, we provide a corresponding upper bound of~$2^{O(n)}$.
Future research includes investigating whether
the following functions~$\lintop_{n,k}$ can be expressed using
DNNF circuits of size~$f(k) \cdot n^{o(k)}$, for some
computable function~$f$.

\begin{definition}
Let~$n,k \in \NN$ with~$n > k$.
The propositional function~$\lintop_{n,k}$
is the Boolean function over the variables~$\SB x_{i,j} \SM i,j \in [n] \SE$
such that~$f[\alpha]$ is true if and only if there exists
some set~$K \subseteq [n]$ with~$|K| = k$ and
some linear order~${\prec} \subseteq K \times K$
such that for each~$i,j \in [n]$
it holds that:
\[ \alpha(x_{i,j}) = \begin{dcases*}
   1 & if~$i,j \in K$ and~$i \prec j$, \\
   0 & if~$i,j \in K$ and~$j \prec i$, \\
   1 & if~$i \in K$ and~$j \not\in K$, \\
   0 & if~$j \in K$ and~$i \not\in K$, and \\
   0 & if~$i,j \not\in K$. \\
\end{dcases*} \]
\end{definition}

We point out the following upper bound,
showing for each~$n,k \in \NN$ with~$n > k$ that we can
express~$\lintop_{n,k}$ using a DNNF circuit
of size~$n^{O(k)}$.

\begin{proposition}
\label{prop:upper-bound-toplin}
For each~$n,k \in \NN$ with~$n > k$, there is a DNNF circuit
of size~$n^{O(k)}$ expressing~$\lintop_{n,k}$.
\end{proposition}
\begin{proof}
Take some~$n,k \in \NN$ with~$n > k$,
and consider~$\lintop_{n,k}$
over the variables~$\SB x_{i,j} \SM i,j \in [n], i < j \SE$.
We construct a DNNF circuit expressing~$\lintop_{n,k}$
as follows.
Let~$T$ be the set of all subsets of~$[n]$ of size at most~$k$---%
that is~$T = \SB S \SM S \subseteq [n], |S| \leq k \SE$.
We know that~$|T| = 1 + \sum\nolimits_{1 \leq i \leq k} \binom{n}{k}$
which is~$O(k n^k)$.
We introduce nodes~$C_{S}$ for each~$S \in T$,
and we introduce
nodes~$C_{i,S}$ for each~$S \in T$ such that~$|S| < k$, and
each~$i \in [n] \setminus S$.
For each~$S \in T$ with~$|S| < k$, we let:
\[ C_{S} = \bigvee\limits_{i \in [n] \setminus S} C_{i,S}. \]
For each~$S \in T$ with~$|S| < k$ and each~$i \in [n] \setminus S$,
we let:
\[ C_{i,S} = \bigwedge\limits_{j \in [n] \setminus S \atop i \neq j}
  x_{i,j} \wedge \neg x_{j,i} \wedge C_{S \cup \SBs i \SEs}. \]
Finally, for each~$S \subseteq [n]$ with~$|S| = k$,
we let:
\[ C_{S} = \bigwedge\limits_{i,j \in [n] \setminus S \atop i \neq j}
  \neg x_{i,j} \wedge \neg x_{j,i}. \]
We then let the node~$C_{\emptyset}$ be the root of the Boolean circuit.
It is straightforward to verify that the circuit
expresses~$\lintop_{n,k}$.
Moreover, the circuit is in DNNF
because each node~$C_{S}$ contains only
variables~$x_{i,j}$ with~$i,j \in [n] \setminus S$.
It is readily verified that the circuit is of size~$n^{O(k)}$---%
there are~$O(kn^k)$ nodes~$C_S$
and~$O(kn^{k+1})$ nodes~$C_{i,S}$.
\end{proof}

\paragraph*{Acknowledgments}
Thanks to Stefan Mengel for pointing out the upper
bound of Proposition~\ref{prop:upper-bound}.

\DeclareRobustCommand{\DE}[3]{#3}

\end{document}